\RequirePackage{fix-cm}
\documentclass[twocolumn]{svjour3}          % twocolumn
\smartqed  % flush right qed marks, e.g. at end of proof

\renewenvironment{proof}{{\em Proof.} }{\hfill $\Box$}
\usepackage{xargs}
\usepackage{hyperref,pdfsync}
\usepackage{aliascnt}
\usepackage{ifthen}
\usepackage{amsmath}
\usepackage{amsfonts}
\usepackage{amssymb}
\usepackage{algorithm}
\usepackage{algorithmicx}
\usepackage{algpseudocode}
\usepackage{enumerate}
\usepackage{stmaryrd}
\usepackage{color}
\usepackage{graphicx}
\usepackage{caption}
\usepackage{subcaption}
\usepackage[textwidth=3cm, textsize=footnotesize]{todonotes}
\newtheorem{thm}{Theorem}

\newaliascnt{defi}{thm}

\aliascntresetthe{defi}

\newaliascnt{lem}{thm}
\newtheorem{lem}[lem]{Lemma}
\aliascntresetthe{lem}

\newaliascnt{prop}{thm}

\aliascntresetthe{prop}

\newaliascnt{rem}{thm}
\newtheorem{rem}[rem]{Remark}
\aliascntresetthe{rem}

\include{dmo2014_def}

\begin{document}

%\title{A statistical perspective on Carlin and Chib's pseudo-prior method\thanks{This work is supported by the Swedish Research Council.}
\title{On the use of Markov chain Monte Carlo methods for the sampling of mixture models\thanks{This work is supported by the Swedish Research Council, Grant 2011-5577.}%Grants or other notes
}
\subtitle{A statistical perspective}


\author{Randal Douc         \and
            Florian Maire         \and
            Jimmy Olsson
}


\institute{R. Douc \at
              Institut T\'el\'ecom/T\'el\'ecom SudParis, CNRS UMR 5157 SAMOVAR, Evry, France \\
%              Tel.: +123-45-678910\\
%              Fax: +123-45-678910\\
              \email{randal.douc@telecom-sudparis.eu}           %  \\
%             \emph{Present address:} 9 rue Charles Fourier, 91011 Evry, France  %  if needed
           \and
           F. Maire \at
           Institut T\'el\'ecom/T\'el\'ecom SudParis, CNRS UMR 5157 SAMOVAR, Evry, France \\
%              Tel.: +123-45-678910\\
%              Fax: +123-45-678910\\
           \email{florian.maire@telecom-sudparis.eu}           %  \\
           \and
           J. Olsson \at
           KTH Royal Institute of Technology, Stockholm, Sweden \\
           \email{jimmyol@kth.se}
}

\date{Received: date / Accepted: date}

\maketitle

\begin{abstract}
In this paper we study asymptotic properties of different data-augmentation-type Markov chain Monte Carlo algorithms sampling from mixture models comprising discrete as well as continuous random variables. Of particular interest to us is the situation where sampling from the conditional distribution of the continuous component given the discrete component is infeasible. In this context, we cast \emph{Carlin \& Chib's pseudo-prior method} into the framework of mixture models and discuss and compare different variants of this scheme. We propose a novel algorithm, the \emph{FCC sampler}, which is less computationally demanding than any Metropolised Carlin \& Chib-type algorithm. The significant gain of computational efficiency is however obtained at the cost of some asymptotic variance. The performance of the algorithm vis-\`a-vis alternative schemes is investigated theoretically, using some recent results obtained in \cite{maire:douc:olsson:2014} for inhomogeneous Markov chains evolving alternatingly according to two different $\starg$-reversible Markov transition kernels, as well as numerically.

%In this paper we study asymptotic properties of different data-augmentation-type Metropolis-Hastings algorithms for sampling mixture models. More specifically, we compare different variants of Carlin and Chib's pseudo-prior method, applied in the context of mixture models. We propose a novel algorithm, named {\em freezed Carlin and Chib-type algorithm}, which is less computationally demanding than any {\em Metropolized Carlin and Chib-type algorithm}, even if it turns out that this algorithm is a bit less efficient in terms of the asymptotic variance. The theoretical comparison is driven using some recent results obtained in \cite{maire:douc:olsson:2014} for inhomogeneous Markov chains that evolve alternatingly according to two different $\starg$-reversible Markov transition kernels. Add sentence that describes examples...

\end{abstract}

\section{Introduction}
\label{sec:introduction}
To sample from \emph{mixture models} $\starg$ comprising a discrete and a continuous random variable, denoted by $M$ and $Z$, respectively, is a fundamental problem in statistics. In this paper we study the use of different data-augmentation-type Markov chain Monte Carlo (MCMC) algorithms for this purpose. Of particular interest to us is the situation where sampling from the conditional distribution of $Z$ given $M$ is infeasible. 

In many applications the most natural approach to sampling from $\starg$ goes via the \emph{Gibbs sampler}, which samples alternatingly from the conditional distributions $M \mid Z$ and $Z \mid M$. Since the component $M$ is discrete, the former sampling step is most often feasible (at least when $\starg$ is known up to a normalising constant). On the contrary,  drawing $Z \mid M$ is in general infeasible; in that case this sampling step is typically \emph{metropolised} by replacing, with a Metropolis-Hastings probability, the value of $Z$ obtained at the previous iteration by a candidate drawn from some proposal kernel. This yields a so-called \emph{Metropolis-within-Gibbs}---or \emph{hybrid}---sampler. 

However, when the modes of the mixture distribution are well-separated, implying a strong correlation between $M$ and $Z$, the Gibbs sampler has in general very limited capacity to move flexibly between the different modes, and exhibits for this reason most often very poor mixing (see \cite{doi:10.1198/1061860031329} for some discussion). Since this problem is due to model dependence, it effects the standard Gibbs as well as the hybrid sampler. In order to cope with this well-known problem, we cast \emph{Carlin \& Chib's pseudo-prior method} \cite{carlin:bayesian} into the framework of mixture models. The method extends the target model with a set of auxiliary variables that are used for moving the discrete component. When the distribution (determined by a set of \emph{pseudo-priors}) of the auxiliary variables is chosen optimally (an idealised situation however), the method produces indeed i.i.d. samples from the marginal distribution of $M$ under $\starg$. Given $M$, the $Z$ component is sampled from $Z \mid M$ in accordance with the Gibbs sampler, with possible metropolisation in the case where exact sampling is infeasible. The latter scheme will be referred as the \emph{Metropolised Carlin \& Chib-type} (MCC) \emph{sampler}. 

Surprisingly, it turns out that passing directly and deterministically the value of the $M$th auxiliary variable, obtained through sampling from the pseudo-priors at the beginning of the loop, to the $Z$ component yields a Markov chain that is still $\starg$-reversible (see \autoref{rem:reversibility}), and using some novel results obtained in \cite{maire:douc:olsson:2014} on the comparison of asymptotic variance for inhomogeneous Markov chains we are able to prove (see \autoref{thm:compMCC-FCC}) that this novel MCMC algorithm, referred to as the \emph{Frozen Carlin \& Chib-type} (FCC) \emph{sampler}, generates a Markov chain whose sample path averages have always higher asymptotic variance than those of the MCC sampler for a large class of objective functions. This is well in line with our expectations, as the MCC sampler ``refreshes'' more often the $Z$ component. On the other hand, since this component is already modified through sampling from the pseudo-priors, which, when well-designed, should be close to the true conditional distribution $Z \mid M$, we may expect that the additional mixing provided by the MCC sampler is only marginal. This is also confirmed by our simulations, which indicate only a small advantage of the MCC sampler to the  FCC sampler in terms of autocorrelation. As the FCC algorithm omits completely the Metropolis-Hastings operation of the MCC sampler, it is considerably more computationally efficient. Thus, we consider the FCC sampler as a strong alternative to the MCC sampler in terms of efficiency (variance per unit CPU). 

The paper is structured as follows: in \autoref{sec:preliminaries} we introduce some notation and describe the mixture model framework under consideration. \autoref{sec:MCMC} describes the Carlin \& Chib-type MCMC samplers studied in the paper. In \autoref{sec:main:results} we prove that the involved algorithms are indeed $\starg$-reversible and provide a theoretical comparison of the MCC and FCC samplers. Finally, in the implementation part, \autoref{sec:numerics}, we illustrate and compare numerically the algorithms on a two exemples: a mixture of Gaussian distributions and a model where the mixture variables are only partially observed.

%and discuss and compare different variants of this scheme. We propose a novel algorithm, the \emph{FCC sampler}, which is less computationally demanding than any Metropolised Carlin \& Chib-type algorithm. The significant gain of computational efficiency is however obtained at the cost of some asymptotic variance. The performance of the algorithm vis-\`a-vis alternative schemes is investigated theoretically, using some recent results obtained in \cite{maire:douc:olsson:2014} for inhomogeneous Markov chains evolving alternatingly according to two different $\starg$-reversible Markov transition kernels, as well as numerically. 

\section{Preliminaries}
\label{sec:preliminaries}
%%%%%%%%%%%%
% Notation
%%%%%%%%%%%%
\subsection{Notation}
\label{subsec:notation}

We assume throughout the paper that all variables are defined on a common probability space $(\Omega, \mathcal{F}, \mathbb{P})$. We will use upper case for random variables and lower case for realisations of the same, and write ``$X \leadsto x$'' when $x$ is realisation of $X$. We write ``$X \sim \mu$'' to indicate that the random variable $X$ is distributed according to the probability measure $\mu$. For any $\mu$-integrable function $h$ we let $\mu(h) \eqdef \int h(x) \mu(\rmd x)$ be the expectation of $h(X)$ under $\mu$. Similarly, for Markov transition kernels $M$Êwe write $Mf(x) \eqdef \int f(x') M(x, \rmd x')$ whenever this integral is well-defined. For any two probability measures $\mu$ and $\mu'$ defined on some measurable spaces $(\Xsp, \Xalg)$ and $(\mathsf{X}', {\Xalg}')$, respectively, we denote by $\mu(\rmd x)\mu'(\rmd x')$ the product measure $\mu \varotimes \mu'(\rmd x \times \rmd x')$ on $(\Xsp \times \mathsf{X}', \Xalg \varotimes \Xalg')$. For $(m, n) \in \zset^2$ such that $m \leq n$, we denote by $\intvect{m}{n} \eqdef \{m, m + 1, \ldots, n\} \subset \zset$. Moreover, we denote by $\Npos \eqdef \nset \setminus \{ 0 \}$ the set of positive integers.

Finally, given some probability measure $\pi$ on $(\Xsp, \Xalg)$ we recall, firstly, that a Markov transition kernel $M$ is called \emph{$\pi$-reversible} if $\pi(\rmd x) M(x, \rmd x') = \pi(x') M(x', \rmd x)$ and, secondly, that $\pi$-reversibility of $M$ implies straightforwardly that this kernel allows $\pi$ as a stationary distribution.
\subsection{Mixture models}
\label{subsec:carlinchib}

Throughout this paper, our main objective is to sample a probability distribution $\starg$ on some product space $\Ysp \eqdef \intvect{1}{n} \times \Zsp$, where $(\Zsp, \Zalg)$ is some (typically uncountable) measurable space, associated with the $\sigma$-field $\Yalg \eqdef 2^{\intvect{1}{n}} \varotimes \Zalg$. Thus, a $\starg$-distributed random variable $\Y{} = (\M{}, \Z{})$ comprises an $\intvect{1}{n}$-valued (discrete) random variable $\M{}$ and a $\Zsp$-valued (typically continuous) random variable $\Z{}$.

In the following we assume that $\starg(\rmd \m \times \rmd \z)$  is dominated by a product measure $|\rmd \m| \nu(\rmd \z)$, where $|\rmd \m|$ denotes the counting measure on $\intvect{1}{n}$ and $\nu$ is some nonnegative measure on $(\Zsp, \Zalg)$, and denote by $\starg(\m, \z)$ the corresponding density function on $\intvect{1}{n} \times \Zsp$. We may then define the conditional density and probability functions
\begin{equation} \label{eq:cond:densities}
	\begin{split}
		\starg(\m \mid \z) &\eqdef \frac{\starg(\m, \z)}{\sum_{\m' = 1}^n \starg(\m', \z)} \eqsp, \\
		\starg(\z \mid \m) &\eqdef \frac{\starg(\m, \z)}{\int \starg(\m, \z') \nu(\rmd \z')}
	\end{split}
\end{equation}
(w.r.t. $|\rmd \m|$ and $\nu$, respectively) on $\intvect{1}{n}$ and $\Zsp$, respectively. We also define the marginal probability function
$$
	\starg(\m) \eqdef \int \starg(\m, \z) \nu(\rmd \z) 
$$
(w.r.t. $|\rmd \m|$) on $\intvect{1}{n}$.

\section{Markov chain Monte Carlo methods for mixture models}
\label{sec:MCMC}
Using the conditional distributions \eqref{eq:cond:densities}, a natural way of sampling $\starg$ consists in implementing a standard Gibbs sampler simulating a Markov chain $\sequence{\Y{\GIBBS}}$ with transitions described by the following algorithm.

\begin{algorithm}[H]
    \begin{algorithmic}[7]
        \caption{Gibbs sampler} \label{alg:algGibbs}
        \Require $\Y[k]{\GIBBS}=(\m, \z)$,
        \begin{enumerate}[(i)]
	        \item draw $\M{}' \sim \starg(\rmd m \mid \z)$ and call the outcome $\m'$ (abbr. $\leadsto \m'$),
	        \item draw $\Z{}' \sim \starg(\rmd z \mid \m') \leadsto \z'$,
	        \item set $\Y[k+1]{\GIBBS}\gets(\m', \z')$.
        \end{enumerate}
    \end{algorithmic}
\end{algorithm}

\begin{rem} \label{rem:metroWithinGibbs}
	Since $\M[]{}$ is a discrete random variable it is most often possible to sample $\M[]{} \sim\starg(\rmd \m \mid \z)$. In contrast, sampling $\Z[]{}\sim\starg(\rmd \z \mid \m)$ is not always possible. In that case, one may replace Step~(ii) by a Metropolis-Hastings step, yielding a Metropolis-within-Gibbs algorithm (see \cite[section 10.3.3]{Robert:MCMC} for details).
\end{rem}

Using the output \autoref{alg:algGibbs}, any expectation $\starg(f)$, where $f$ is some $\starg$-integrable objective function on $\Ysp$, can be estimated by the sample path average
$$
\hat{\pi}^{(\GIBBS)}_n(f) \eqdef \frac{1}{n} \sum_{k = 0}^{n - 1} f(\Y[k]{\GIBBS}) \eqsp.
$$
Even though \autoref{alg:algGibbs} generates a Markov chain $\sequence{\Y{\GIBBS}}$ with stationary distribution $\starg$, the discrete component $\sequence{\M{\GIBBS}}$ tends in practice to get stuck in a few states. Indeed, when the variable $Z$ is sampled from its conditional distribution given $M = \m$, the probability of jumping to another index $m' \neq m$ is proportional to $\starg(m', \z)$, which may be very low when the index component $M$ is informative concerning the localisation of $Z$. This will lead to poor mixing and, consequently, high variance of any estimator $\hat{\pi}^{(\GIBBS)}_n(f) $. To be specific, let $h$ be some function on $\intvect{1}{n}$ and assume that we run the Gibbs sampler in \autoref{alg:algGibbs} to estimate $\int h(m) \starg(\rmd m)$. Then $\sequence{\M{\GIBBS}}$ is itself a Markov chain with transition kernel $G$, say, and, starting with $\M[0]{\GIBBS} \sim \starg(\rmd m)$, we have
\begin{multline*}
\covar{h(\M[0]{\GIBBS})}{h(\M[1]{\GIBBS})}=\covar{h(\M[0]{\GIBBS})}{Gh(\M[0]{\GIBBS})}\\
=\int \starg(\rmd z) \left(\int  \starg(\rmd m \mid z ) h(m)\right)^2 \geq 0 \eqsp. 
\end{multline*}
Combining this with the fact that $G$ is $\starg(\rmd m)$-reversible, we obtain
\begin{multline*}
\covar{h(\M[0]{\GIBBS})}{h(\M[2k+1]{\GIBBS})}\\
=\covar{G^k h(\M[0]{\GIBBS})}{G (G^kh)(\M[0]{\GIBBS})} \geq 0 \eqsp.
\end{multline*}
Moreover, using again that $G$ is $\starg(\rmd m)$-reversible,
\begin{multline*}
\covar{h(\M[0]{\GIBBS})}{h(\M[2k]{\GIBBS})}\\
=\covar{G^k h(\M[0]{\GIBBS})}{G^kh(\M[0]{\GIBBS})} \geq 0  \eqsp. 
\end{multline*}
Finally, letting $f(m,z) \equiv h(m)$, we obtain 
$$
\var{\sqrt{n} \hat{\pi}^{(\GIBBS)}_n(f)} \geq \var{h(\M[0]{\GIBBS})} \eqsp, 
$$
showing that the Gibbs sampler approximates the index less accurately than i.i.d. sampling from $\starg(\rmd m)$.
%\end{rem}

The \emph{pseudo-prior method} of B.~P.~Carlin and S.~Chib \cite{carlin:bayesian} was introduced in the context of model selection and can successfully be adapted to mixture models. By introducing some auxiliary variables, this method increases the number of moves of the index component. The algorithm may be regarded as a Gibbs sampler-based data-augmentation algorithm targeting the distribution $\targ$ defined on the extended state space $\intvect{1}{n} \times \Zsp^n$ by
%\begin{equation} \label{eq:pi_CC}
%	\targ(\rmd \m \times \rmd\auxi) \eqdef \starg(\rmd \m \times \rmd \auxi_\m) \prod_{j \neq \m} \rho_j(\rmd \auxi_j) \eqsp,
%\end{equation}
\begin{equation} \label{eq:pi_CC}
	\targ(\rmd \m \times \rmd\auxi) \eqdef \starg(\rmd \m \times \rmd \auxi_\m) \bigotimes_{j \neq \m} \rho_j(\rmd \auxi_j) \eqsp,
\end{equation}
where $\auxi = (\auxi_1, \ldots, \auxi_n) \in \Zsp^n$ and the probability measures $\{\rho_j ; j \in \intvect{1}{n} \}$ are referred to as \emph{pseudo-priors} or \emph{linking densities} (the terminology comes from \cite{carlin:bayesian}). We assume that also the pseudo-priors
are dominated jointly by the nonnegative measure $\nu$ and use the same symbols $\rho_j$ for denoting the corresponding densities. As a consequence, also the measure $\targ$ defined in \eqref{eq:pi_CC} has a density
\begin{equation*} %\label{eq:pi_CC}
	\targ(\m, \auxi) \eqdef \starg(\m, \auxi_\m) \prod_{j \neq \m} \rho_j(\auxi_j) \quad ((\m, u) \in \intvect{1}{n} \times \Zsp^n)
\end{equation*}
with respect to $|\rmd \m| \varotimes \nu^{\varotimes n}(\rmd u)$.

The choice of the pseudo-priors is tuned by the user provided that these are analytically tractable and can be sampled from. Denote by $\sequence{\Y{\CC}}$ the Markov chain generated by this algorithm, which we will in the following refer to as the \emph{Carlin \& Chib-type} (CC-type) \emph{sampler} and whose transitions comprise the $n + 1$ sub-steps  described in \autoref{alg:algCC} below.

\begin{algorithm}[H]
    \begin{algorithmic}[7]
        \caption{CC-type sampler} \label{alg:algCC}
        \Require $\Y[k]{\CC} = (\m, \auxi_m)$,
        \begin{enumerate}[(i)]
        \item for all $j \neq \m$, draw $\Auxi_j \sim \rho_j \leadsto \auxi_j$,
        \item draw $\M{}' \sim \targ(\rmd \m \mid \auxi) \leadsto \m'$,
        \item \label{item:algCC:piStar}draw $\Auxi_{m'}' \sim \starg(\rmd \z \mid \m') \leadsto \auxi'_{\m'}$,
        \item set $\Y[k + 1]{\CC}\gets (\m', \auxi'_{\m'})$.
        \end{enumerate}
    \end{algorithmic}
\end{algorithm}

Intuitively, \autoref{alg:algCC} allows different models to be visited more frequently than in the Gibbs sampler; indeed, in Step~(ii) the probability of moving to the index $m'$ is
\begin{equation} \label{eq:CC_poids}
	\targ(m' \mid \auxi) \propto \starg(m',\auxi_{m'})/ \rho_{m'}(\auxi_{m'})\eqsp, 
\end{equation}
where the right hand side is close to $\starg(m')$ if the pseudo-priors are chosen such that $\rho_{\ell}(\z)$ is close to $\starg(\z \mid \ell)$ for all $(\ell, \z) \in \intvect{1}{n} \times \Zsp$. The optimal case where $\rho_\ell(\z) \equiv \starg(\z \mid \ell)$ implies, via \eqref{eq:CC_poids}, that
$$
	\targ(m' \mid \auxi)\propto\starg(m')\eqsp.
$$
Thus, in this case, Step~(ii) draws actually $\M{}'$ according to the \emph{exact} marginal $\starg(\rmd \m)$ of the class index random variable regardless the value of $\auxi$, which implies that the algorithm simulates \emph{i.i.d. samples} according to $\starg$. This actually gives a more efficient approximation than that produced by the Gibbs sampler, whose variance w.r.t. the index component is, as we remarked previously, always larger than that obtained through i.i.d. sampling from $\starg(\rmd \m)$. However, this ideal situation requires the quantity $\starg(\z \mid \ell)$ to be tractable, which is typically not the case.
%\end{rem}

As in \autoref{rem:metroWithinGibbs}, one may replace Step~(iii) in \autoref{alg:algCC} by a Metropolis-Hastings step if sampling from $\starg(\rmd \z \mid \m')$ is infeasible. This is most often the case when $\starg(\rmd \m \times \rmd \z)$ is the {\em a posteriori} distribution of $(M, Z)$ conditionally on one or several observations (see \autoref{sec:partial:observations} for an example). In that case $\starg$ is known only up to a normalizing constant, which prevents sampling from the conditional density $\starg(\rmd \z \mid \m')$. The resulting algorithm will in the following be referred to as the Metropolised CC-type (MCC) sampler and is presented in \autoref{alg:algMCC}, where $\{ÊR_\ell ; \ell \in \intvect{1}{n} \}$ is set of proposal kernels on $\Zsp \times  \Zalg$. Assume for simplicity that all these kernels are jointly dominated by the reference measure $\nu$ and denote by $\{Êr_\ell ; \ell \in \intvect{1}{n} \}$ the corresponding transition densities with respect to this measure. Introducing also the Metropolis-Hastings acceptance probability
\begin{equation} \label{eq:def:alpha}
	\alpha_\ell(u, z) \eqdef \frac{\starg(\ell, z) r_\ell(z, u)}{\starg(\ell, u) r_\ell(u, z)} \wedge 1 \quad ((\ell, u, z) \in \intvect{1}{n} \times \Zsp^2) \eqsp,
\end{equation}
the MCC algorithm is described as follows.
\begin{algorithm}[H]
    \begin{algorithmic}[7]
        \caption{MCC sampler} \label{alg:algMCC}
        \Require $\Y[k]{\MCC} = (\m, \auxi_m)$,
        \begin{enumerate}
        \item [(i)] for all $j \neq \m$, draw $\Auxi_j \sim \rho_j \leadsto \auxi_j$,
        \item [(ii)] draw $\M{}' \sim \targ(\rmd \m \mid \auxi) \leadsto \m'$,
        \item [(iii.1)] draw $Z \sim R_{m'}(\auxi_{m'}, \rmd z) \leadsto z$,
        \item [(iii.2)] set $\Auxi_{m'}' \gets
        \begin{cases}
         z& \mbox{w.~pr. $\alpha_{m'}(u_{m'}, z)$} \eqsp, \\
        \auxi_{\m'} & \mbox{otherwise} \eqsp,
        \end{cases} \leadsto  \auxi'_{\m'}$ \eqsp,
        \item [(iv)] set $\Y[k + 1]{\MCC} \gets (\m', \auxi'_{\m'})$.
        \end{enumerate}
    \end{algorithmic}
\end{algorithm}
Note that Step~\eqref{item:algCC:piStar} generates, given $\auxi_{m'}$, $\Auxi_{m'}' \sim K_{m'}(\auxi_{m'}, \rmd u')$, where
\begin{multline} \label{eq:def:K-m}
K_\ell(u, \rmd u') \eqdef R_\ell(u, \rmd u') \alpha_\ell(u, u') \\Ê
+\delta_u(\rmd u') \left( 1 - \int R_\ell(u, \rmd u'') \alpha_\ell(u, u'') \right) \\Ê
((u, \ell) \in \intvect{1}{n}) \eqsp.
\end{multline}
%\begin{equation}  \label{eq:def:K-m}
%K_{m'}(x,A \setminus \{x\})=\int_{A \setminus \{x\}} R_{m'}(x, \rmd x') \alpha_{m'}(x,x')\eqsp.
%\end{equation}
It can be easily checked (using \eqref{eq:def:alpha}) that $K_{m'}$ is indeed a Metropolis-Hastings kernel with respect to $\pi(\rmd z \mid m')$; it is thus $\pi(\rmd z \mid m')$-reversible.

Remarkably, it turns out that Step~\eqref{item:algCC:piStar} in \autoref{alg:algMCC} may be omitted, which may, in some cases, imply a significant gain of computational complexity. The MCC sampler then simplifies to what we will refer to as the Frozen CC-type (FCC) sampler which is described formally as follows.
\begin{algorithm}[H]
    \begin{algorithmic}[7]
        \caption{FCC sampler} \label{alg:algFCC}
        \Require $\Y[k]{\FCC} = (\m, \auxi_m)$,
        \begin{enumerate}[(i)]
        \item for all $j \neq \m$, draw $\Auxi_j \sim \rho_j \leadsto \auxi_j$,
        \item draw $\M{}' \sim \targ(\rmd \m \mid \auxi) \leadsto \m'$,
        \item $\auxi'_{\m'}\gets\auxi_{\m'}$,
        \item set $\Y[k + 1]{\FCC} \gets(\m', \auxi'_{\m'})$.
        \end{enumerate}
    \end{algorithmic}
\end{algorithm}
As remarked in \autoref{rem:reversibility} in the next section, this novel algorithm produces a Markov chain $\{\Y[k]{\FCC} ; k \in \nset\}$ that indeed admits $\starg$ as an invariant distribution. Nevertheless, as stated in \autoref{thm:compMCC-FCC} below, this algorithm is always less efficient in terms of asymptotic variance than the corresponding MCC sampler. Intuitively, this stems from the fact that once the index $\M{}'$ is drawn, the associated continuous component is selected deterministically without being ``refreshed'' (on the contrary to Step~\eqref{item:algCC:piStar} in \autoref{alg:algMCC}). Nevertheless, as mentioned above, \autoref{alg:algFCC} skips completely the Metropolis-Hastings step (Step~\eqref{item:algCC:piStar}) of \autoref{alg:algMCC} and is therefore considerably less demanding from a computational point of view. In addition, as our numerical simulations indicate that the gain of the asymptotic variance obtained by refreshing, as in the MCC sampler, this component instead of freezing the same as in the FCC sampler seems to be limited  (see \autoref{sec:numerics} for details), we definitely regard the FCC algorithm as a strong challenger of the MCC sampler.

\section{Theoretical results}
\label{sec:main:results}
\subsection{Comparison of asymptotic variance of inhomogeneous Markov chains}

In this section we recall briefly the main result of \cite[Theorem~4]{maire:douc:olsson:2014}, which is propelling the coming analysis. The following---now classical---orderings of Markov kernels turns out to be highly useful.
\begin{definition} \label{defi:Peskun:ordering}
	Let $\P[0]$ and $\P[1]$ be Markov transition kernels on some state space $(\Xsp, \Xalg)$ with common invariant distribution $\targ$. We say that $\P[1]$ \emph{dominates} $\P[0]$
	\begin{itemize}
	\item \emph{on the off-diagonal}, denoted $\P[1] \succeq \P[0]$, if for all $\set{A} \in \Xalg$ and $\targ$-a.s. all $x \in \Xsp$,
$$
	\P[1](x, \set{A} \setminus \{ x \}) \geq \P[0](x, \set{A} \setminus \{ x \}) \eqsp.
$$	\item \emph{in the covariance ordering}, denoted $\P[1] \pgeq[1] \P[0]$, if for all $f \in \Ltwo[\targ]$,
$$
	\int f(x) \P[1]f(x) \targ(\rmd x) \leq \int f(x) \P[0]f(x) \targ(\rmd x) \eqsp.
$$
%	$$
%		\pscal{f}{\P[1]f} \leq \pscal{f}{\P[0]f} \eqsp.
%	$$
	\end{itemize}
\end{definition}
The covariance ordering, which was introduced implicitly in \cite[p.~5]{tierney:note} and formalised in \cite{mira:ordering}, is an extension of the off-diagonal ordering, since, according to \cite[Lemma 3]{tierney:note}, $\P[1] \succeq \P[0]$ implies $\P[1] \pgeq[1] \P[0]$. Moreover, it turns out that for reversible kernels, $\P[1] \pgeq[1] \P[0]$ implies that the asymptotic variance of sample path averages of chains generated by $\P[1]$ is smaller than or equal to that of chains generated by $\P[0]$ (see the proof of \cite[Theorem~4]{tierney:note}).

%In this section we apply \autoref{thm:mainResult} in the context of Gibbs samplers.
In algorithms of Gibbs-type, the ordering in \autoref{defi:Peskun:ordering} is usually not applicable, since the fact that all candidates are accepted with probability one prevents the chain from remaining in the same state. The ordering is however still meaningful when a component is discrete.%, in which case the analysis can again be cast into the framework of \autoref{thm:mainResult}.

In the following, let $\P[i]$ and $\Q[i]$, $i \in \intvect{0}{1}$, be Markov transition kernels on $(\Xsp, \Xalg)$ and let $\sequence{\X{0}}$ and $\sequence{\X{1}}$ be inhomogeneous Markov chains evolving as follows:
\begin{equation} \label{eq:eq1:Markov}
	\X[0]{i} \stackrel{\P[i]}{\longrightarrow} \X[1]{i} \stackrel{\Q[i]}{\longrightarrow} \X[2]{i} \stackrel{\P[i]}{\longrightarrow} \X[3]{i} \stackrel{\Q[i]}{\longrightarrow} \cdots
\end{equation}
This means that for all $k \in \nset$ and $i \in \{0, 1\}$,%, and $\set{A} \in \Xalg$,
\begin{itemize}
	\item $\mathbb{P}\left( \X[2 k+1]{i} \in \rmd x %\set{A}
	\mid \mathcal{F}_{2 k}^{(i)} \right) = \P[i](\X[2 k]{i}, \rmd x %\set{A}
	)$,
	\item $\mathbb{P}\left( \X[2 k+2]{i} \in \rmd x %\set{A}
	\mid \mathcal{F}_{2 k+1}^{(i)} \right) =  \Q[i](\X[2 k+1]{i}, \rmd x %\set{A}
	)$,
\end{itemize}
where $\mathcal{F}_{n}^{(i)} \eqdef \sigma(\X[0]{i},\ldots, \X[n]{i})$, $n \in \nset$. Now, impose the following assumption.
\begin{hypA} \label{ass:Peskun:order}
	\begin{enumerate}[$(i)$]
		\item $\P[i]$ and $\Q[i]$, $i \in \intvect{0}{1}$, are $\targ$-reversible,
		\item $\P[1] \pgeq[1] \P[0]$ and $\Q[1] \pgeq[1] \Q[0]$.
	\end{enumerate}
\end{hypA}
As mentioned above, $\P[1]\succeq \P[0]$ implies $\P[1]\pgeq[1]\P[0]$; thus, in practice, a sufficient condition for {\A{ass:Peskun:order}}\,(ii) is that $\P[1] \succeq \P[0]$ and $\Q[1] \succeq \Q[0]$. Under these assumptions,  \cite{maire:douc:olsson:2014} established the following result.
\begin{thm}[\cite{maire:douc:olsson:2014}]%[{\cite[Theorem~4]{maire:douc:olsson:2014}}]
	\label{thm:inhomo:ordering}
	Assume that $\P[i]$ and $\Q[i]$, $i \in \intvect{0}{1}$, satisfy {\A{ass:Peskun:order}} and let $\sequence{\X{i}}$, $i \in \intvect{0}{1}$, be Markov chains evolving as in \eqref{eq:eq1:Markov} with initial distribution $\targ$. Then for all $f \in \Ltwo[\targ]$ such that for $i \in \intvect{0}{1}$,
	\begin{multline} \label{eq:assumpFuncThm}
		\sum_{k = 1}^{\infty} \left( |\covar{f(\X[0]{i})}{f(\X[k]{i})}| \right. \\
	\left. +|\covar{f(\X[1]{i})}{f(\X[k+1]{i})}| \right) < \infty
	\end{multline}
	it holds that
	\begin{equation} \label{eq:main:result}
		\asvar{1}{f} \leq \asvar{0}{f} \eqsp,
	\end{equation}
	where
	\begin{equation} \label{eq:main:result:limits}
		\asvar{i}{f} \eqdef \lim_{n \to \infty} \frac{1}{n} \var{\sum_{k = 0}^{n - 1}f(\X[k]{i})} \quad (i \in \intvect{0}{1}) \eqsp.
	\end{equation}
\end{thm}

\begin{rem}
As shown in \cite[Proposition~9]{maire:douc:olsson:2014}, under the assumption that the product kernels $\P[i] \Q[i]$, $i \in \intvect{0}{1}$, are both \emph{$V$-geometrically ergodic} (according to Definition~7 in the same paper), the absolute summability assumption \eqref{eq:assumpFuncThm} holds true for all objective functions $f$ such that $f$ and $\P[i]f$, $i \in \intvect{0}{1}$, have all bounded $\sqrt{V}$-norm; see again \cite{maire:douc:olsson:2014} for details.
\end{rem}

\subsection{The MCC sampler vs. the FCC sampler}
In the light of the remarks following \autoref{alg:algFCC} it is reasonable to assume that the CC-type sampler (\autoref{alg:algCC}) and the MCC sampler provides more accurate estimates than the FCC sampler. However, since neither $\sequence{\Y{\MCC}}$ nor $\sequence{\Y{\FCC}}$ are $\starg$-reversible, \cite[Theorem 4]{tierney:note} does not allow these two algorithms to be compared. Nevertheless, using \autoref{thm:inhomo:ordering} we may provide a theoretical justification advocating the MCC and CC-type samplers ahead of the FCC sampler in terms of asymptotic variance. To do this we first embed $\sequence{\Y{\MCC}}$ and $\sequence{\Y{\FCC}}$ into inhomogeneous $\starg$-reversible Markov chains $\sequence{\X{\MCC}}$ and $\sequence{\X{\FCC}}$ defined on $\Ysp = \intvect{1}{n} \times \Zsp$ through, for $\alg \in \intvect{\MCC}{\FCC}$:
\begin{multline} \label{alg:cc:gibbs}
	\X[2k]{\alg}=\begin{pmatrix}\M[k]{\alg}\\ \Z[k]{\alg}
	\end{pmatrix} \stackrel{\P[\alg]}{\longrightarrow} \X[2k+1]{\alg}=\begin{pmatrix}\cM[k+1]{\alg}\\ \cZ[k+1]{\alg}
	\end{pmatrix}
	\\Ê\stackrel{\Q[\alg]}{\longrightarrow} \X[2k+2]{\alg}=\begin{pmatrix}\M[k+1]{\alg} \\ \Z[k+1]{\alg}
	\end{pmatrix} \stackrel{\P[\alg]}{\longrightarrow} \cdots
\end{multline}
Here we have defined the kernels
\begin{itemize}%[(i)]
\item
$
\displaystyle \P[\MCC]((m, z), \rmd \check{m} \times \rmd \check{z}) \\Ê
\eqdef \idotsint \left( \prod_{j \neq m} \rho_j(\rmd u_j) \right) \delta_z(\rmd u_m) \targ(\rmd \check{m} \mid u) \delta_{u_{\check{m}}}(\rmd \check{z})
$,

\item $\P[\FCC] \eqdef  \P[\MCC]$,
\item $\Q[\MCC]((\check{m}, \check{z}), \rmd m \times \rmd z) \eqdef \delta_{\check{m}}(\rmd m) K_{\check{m}}(\check{z},\rmd z) $ (where $K_{\cm}$ is defined in \eqref{eq:def:K-m}),
\item $\Q[\FCC] ((\check{m}, \check{z}), \rmd m \times \rmd z) \eqdef \delta_{\check{m}}(\rmd m)\delta_{\check{z}}(\rmd z)$.
\end{itemize}
Setting $\Y[k]{\alg} \eqdef (\M[k]{\alg}, \Z[k]{\alg})$, $k \in \nset$, $\alg \in \intvect{\GIBBS}{\CC}$, it can be checked easily that $\sequence{\Y{\MCC}}$ and $\sequence{\Y{\FCC}}$ have indeed exactly the same distribution as the output of \autoref{alg:algMCC} and \autoref{alg:algFCC}, respectively.

\begin{thm} \label{rem:reversibility}
The Markov chains  $\sequence{\Y{\MCC}}$ and $\sequence{\Y{\FCC}}$ generated  by \autoref{alg:algMCC} and \autoref{alg:algFCC}, respectively, have $\starg$ as invariant distribution.
\end{thm}

\begin{proof}
The result is established by noting that $\Q[\FCC]$ defined above is reversible with respect to any distribution, and in particular it is $\starg$-reversible. Moreover, according to \autoref{lem:rev:MCC} (below), $\P[\MCC]=\P[\FCC]$ and $\Q[\MCC]$ are also $\starg$-reversible. The statement of the theorem follows.
\end{proof}

\begin{thm}
\label{thm:compMCC-FCC}
Let \sequence{\X{\alg}}, $\alg \in \intvect{\MCC}{\FCC}$, be the Markov chains  \eqref{alg:cc:gibbs} starting with $\X[0]{\alg} \sim \starg$ for $\alg \in \intvect{\MCC}{\FCC}$. Then for all real-valued functions $h$ such that for $\alg \in \intvect{\MCC}{\FCC}$,
$$
\sum_{k = 1}^{\infty} |\covar{h(\M[0]{\alg})}{h(\M[k]{\alg})}| < \infty
$$
it holds that
$$
\lim_{n\to\infty} \frac{1}{n}\var{\sum_{k=1}^{n}h(\M[k]{\MCC})} \leq \lim_{n\to\infty} \frac{1}{n}\var{\sum_{k=1}^{n}h(\M[k]{\FCC})} \eqsp.
$$
\end{thm}

\begin{proof}
By \autoref{rem:reversibility}, the processes $\{\X[k]{\GIBBS} ;  k \in \nset\}$ and $\{\X[k]{\CC} ; k \in \nset\}$ are both inhomogeneous Markov chains that evolve alternatingly according to the $\starg$-reversible kernels $\P[\alg]$ and $\Q[\alg]$, $\alg\in\intvect{\GIBBS}{\CC}$. Moreover, $\P[\MCC]=\P[\FCC] \succeq \P[\FCC]$, and since $\Q[\FCC]$ has no off-diagonal component, $\Q[\MCC] \succeq \Q[\FCC]$. Now, define Markov chains $\sequence{\X{i}}$, $i \in \intvect{\MCC}{\FCC}$, as in \eqref{alg:cc:gibbs} with $\X[0]{i} \sim \targ$, and set $f(m,z) \equiv h(m)$. By construction, $\cM[k]{i}=\M[k]{i}$ for $i \in \intvect{\MCC}{\FCC}$ and $k \in \Npos$,  implying that
\begin{multline} \label{eq:implied:bddness}
	\sum_{k = 1}^{\infty} \left( |\covar{f(\X[0]{i})}{f(\X[k]{i})}| + |\covar{f(\X[1]{i})}{f(\X[k+1]{i})}| \right) \\
	= \targ f^2 - \targ^2 f + 4 \sum_{k = 1}^\infty |\covar{h(\M[0]{i})}{h(\M[k]{i})}| < \infty \eqsp.
\end{multline}
Moreover, for all $n \in \Npos$ and $i \in \intvect{\MCC}{\FCC}$,
\begin{multline*}
		\var{\sum_{k = 1}^n h(\M[k]{i})} = \var{\sum_{k = 1}^n h(\cM[k]{i})} \\
= \frac{1}{4} \var{\sum_{k =1}^{2n} f(\X[k]{i})} \eqsp,
\end{multline*}
which implies, by \eqref{eq:implied:bddness}, that for $i \in \intvect{\MCC}{\FCC}$,
\begin{equation*}
\lim_{n\to\infty} \frac{1}{n}\var{\sum_{k=1}^n h(\M[k]{i})} = \frac{1}{2} \lim_{n \to \infty} \frac{1}{n} \var{\sum_{k = 1}^n f(\X[k]{i})} \eqsp.
\end{equation*}
Finally, by \eqref{eq:implied:bddness} we may apply \autoref{thm:inhomo:ordering} to the chains $\{ \X[k]{i} ; k \in \nset \}$, $i \in \intvect{0}{1}$, which establishes immediately the statement of the theorem.

\end{proof}

\begin{lem} \label{lem:rev:MCC}
The Markov kernels $\P[\MCC]$ and $\Q[\MCC]$ are both $\starg$-reversible.
\end{lem}

\begin{proof}
Write, using the identity
\begin{multline*}
\nu(\rmd\z) \delta_{\z}(\rmd\auxi_m) \delta_{\auxi_{\cm}}(\rmd \cz) \prod_{j \neq m} \nu(\rmd\auxi_j) \\Ê
= \delta_{\auxi_m}(\rmd\z) \delta_{\auxi_{\cm}}(\rmd\cz) \prod_{j = 1}^n \nu(\rmd \auxi_j) \eqsp,
\end{multline*}
for any nonnegative measurable function $f$ on $(\Ysp,\Yalg)$,
\[
\begin{split} %\label{eq:CC_rev_1}
\lefteqn{\iint f(\y, \cy) \starg(\rmd\y) \P[\MCC](\y, \rmd\cy)} \\
=& \idotsint f(\y, \cy) \starg(\m, \z) |\rmd \m| \nu(\rmd \z) \delta_{\z}(\rmd \auxi_m) \\
&\qquad \times \left(
\prod_{j \neq \m} \rho_j(\rmd \auxi_j)
\right)
\frac{\starg(\cm, \auxi_{\cm}) \prod_{j \neq \cm} \rho_j(\auxi_j)}{\sum_{k = 1}^n \starg(k, \auxi_k) \prod_{\ell \neq k} \rho_\ell(\auxi_{\ell})} \\
&\qquad%\hspace{20mm}
\times |\rmd \cm|Ê\delta_{\auxi_{\cm}}(\rmd \cz) \\
=& \idotsint f(\y, \cy) \starg(\m, \z) \starg(\cm,\auxi_{\cm}) \\
&\qquad%\hspace{20mm}
\times \frac{\prod_{j\neq m}\rho_j(\auxi_j) \prod_{j\neq \cm}\rho_j(\auxi_j)}{\sum_{k = 1}^n \starg(k, \auxi_k) \prod_{\ell \neq k} \rho_{\ell}(\auxi_\ell)} \\
&\qquad%\hspace{20mm}
\times \left( \prod_{j = 1}^n \nu(\rmd \auxi_j) \right) |\rmd \m| |\rmd \cm| \delta_{\auxi_{\m}}(\rmd\z) \delta_{\auxi_{\cm}}(\rmd\cz) \eqsp.
\end{split}
\]
Thus, integrating first over $\z$ and $\cz$ and defining
\begin{multline*}
A(\m, \cm, \auxi) \\Ê
\eqdef \starg(\m, \auxi_{\m}) \starg(\cm,\auxi_{\cm})
\frac{\prod_{j \neq m} \rho_j(\auxi_j) \times \prod_{j\neq \cm} \rho_j(\auxi_j)}{\sum_{k = 1}^n\starg(k, \auxi_k) \prod_{\ell \neq k} \rho_{\ell}(\auxi_\ell)}
\end{multline*}
yields
\[
\begin{split}
\lefteqn{\iint f(\y, \cy) \starg(\rmd\y) \P[\CC](\y, \rmd\cy)} \\
&= \idotsint f((\m, \auxi_{\m}), (\cm,\auxi_{\cm})) A(\m, \cm, \auxi) \\
& \hspace{20mm} \times \prod_{j = 1}^n \nu(\rmd \auxi_j) |\rmd \m| |\rmd \cm| \eqsp.
\end{split}
\]
Now, the symmetry $A(m, \cm, \auxi) = A(\cm, m, \auxi)$ implies the identity
\begin{equation} \label{eq:CC:rev:id}
\iint f(\y, \cy) \starg(\rmd\y) \P[\MCC](\y, \rmd\cy) = \iint f(\y, \cy) \starg(\rmd\cy) \P[\MCC](\cy, \rmd\y) \eqsp,
\end{equation}
and as $f$ was chosen arbitrarily, \eqref{eq:CC:rev:id} implies that
$$
\starg(\rmd\y) \P[\MCC](\y, \rmd\cy) = \starg(\rmd\cy) \P[\MCC](\cy, \rmd\y) \eqsp,
$$
which establishes the $\starg$-reversibility of $\P[\MCC]$.

We show that $\Q[\MCC]$ is $\starg$-reversible. Again, let $f$ be some nonnegative measurable function on $(\Ysp,\Yalg)$. Then, using that $K_m$ is reversible with respect to $\starg(\rmd z \mid m)$ for all $m \in \intvect{1}{n}$, we obtain, denoting $\cy \eqdef (\cm, \cz)$ and $y \eqdef (m, z)$,
\[
\begin{split}
\lefteqn{\idotsint f(\cy, y) \starg(\rmd \cy) K_{\cm}(\cz,\rmd z) \delta_{\cm}(\rmd m)} \\
& =\idotsint f(\cy, y) \starg(\rmd \cm) \starg(\rmd \cz \mid \cm) K_{\cm}(\cz,\rmd z) \delta_{\cm}(\rmd m) \\
& =\idotsint f(\cy, y) \starg(\rmd \cm) \starg(\rmd z \mid \cm) K_{\cm}(z,\rmd \cz) \delta_{\cm}(\rmd m) \\
& =\idotsint f(\cy, y) \starg(\rmd m)\starg(\rmd z \mid m) K_m(z,\rmd \cz) \delta_m (\rmd \cm) \eqsp.
\end{split}
\]
This implies
\begin{multline*}
\iint f(\cy, \y) \starg(\rmd \y) \Q[\MCC](\y, \rmd\cy) \\Ê
= \iint f(\cy, \y) \starg(\rmd\cy) \Q[\MCC](\cy, \rmd \y) \eqsp,
\end{multline*}
which completes the proof.
\end{proof}

\section{Numerical illustrations}
\label{sec:numerics}
In this section we compare numerically the performances of the different algorithms described in the previous section. The comparisons will be based on two different models: firstly, a simple toy model consisting of a mixture of two Gaussian strata and, secondly, a model where only partial observations of the mixture variables are available. All implementations are in \textsc{Matlab}, running on a MacBook Air with a $1.8$ GHz Inter Core i7 processor.

\subsection{Mixture of Gaussian strata}
\label{example:toy}

Let $\Ysp = \intvect{1}{2} \times \rset$ (i.e. $\Zsp = \rset$ in this case) and consider a pair of random variables $(M, Z)$ distributed according to the Gaussian mixture model
\begin{equation} \label{eq:toy:model}
\starg(\m, \z) = \frac{1}{2} \normd{z}{\mu_m}{\sigma^2} \quad ((m, z) \in \Ysp) \eqsp,
\end{equation}
where $\sigma > 0$, $(\mu_1, \mu_2) = (-1, 1)$, and $\normd{z}{\mu}{\sigma^2}$ denotes the Gaussian probability density function with mean $\mu$ and variance $\sigma^2$. Even though it is straightforward to generate i.i.d. samples from this simple toy model, we use it for illustrating and comparing the performances of the algorithms proposed in the previous; in particular, since $\starg(\z \mid \m)$ is simply a Gaussian distribution in this case, it is possible execute Step~(iii) in \autoref{alg:algCC} (which is, as mentioned, far from always the case; see the next example). For small values of $\sigma$, such as  the value $\sigma = \sqrt{.2}$ used in this simulation, the two modes are well-separated, implying a strong correlation between the discrete and continuous components. As a consequence we may expect the naive Gibbs sampler to exhibit a very sub-optimal performance in this case. In order to improve mixing we introduced Gaussian pseudo-priors
$$
\rho_\ell(\auxi) \eqdef \normd{\auxi}{\tilde{\mu}_\ell}{\tilde{\sigma}^2_\ell}\quad ((\ell, \auxi) \in \intvect{1}{2} \times \rset) \eqsp
$$
on $\rset$, where $(\tilde{\mu}_1, \tilde{\mu}_2, \tilde{\sigma}_1^2, \tilde{\sigma}_2^2) = (-.5, .5, .15, .25)$, and executed, using these pseudo-priors, \autoref{alg:algCC}, \autoref{alg:algMCC}, and \autoref{alg:algFCC}. Moreover, the naive Gibbs sampler was implemented for comparison. \autoref{alg:algMCC} used the proposal
$$
R_\ell(\auxi, \rmd \z) = \rho_\ell(\rmd z) \quad ((\ell, \auxi) \in \intvect{1}{2} \times \rset) \eqsp,
$$
yielding an algorithm that can be viewed as a hybrid between  \autoref{alg:algCC} and \autoref{alg:algFCC} in the sense that it ``refreshes'' randomly the continuous component $\Auxi_{\m'}$ obtained after Step~(ii) by replacing, with the  Metropolis-Hastings probability $\alpha_{\m'}$, the same by a draw from $\rho_\m'$. Cf. \autoref{alg:algCC} and \autoref{alg:algFCC}, where $\Auxi_{\m'}$ is refreshed systematically according to $\rho_\m'$ and kept frozen, respectively. For each of these algorithms we generated an MCMC trajectory comprising $101,\!000$ iterations (where the first $1,\!000$ iterations were regarded as burn-in and discarded) and estimated the corresponding autocorrelation functions. The outcome, which is displayed in \autoref{fig:autocorrelation:toy:model} below, indicates increasing autocorrelation for the CC, MCC, FCC, and Gibbs algorithms, respectively, confirming completely the theoretical results obtained in the previous section. Interestingly, the FCC algorithm has, despite being close to twice as efficient in terms of CPU with our implementation, only slightly higher autocorrelation than the MCC algorithm (the same applies to both the components). As expected, the Gibbs sampler suffers from very large autocorrelation as it tends to get stuck in the different modes, while the CC algorithm has the highest performance at a computational complexity that is comparable to that of the FCC algorithm in this case (due to \textsc{Matlab}'s very efficient Gaussian random number generator). Qualitatively, similar outcomes are obtained if the parametrisations of the target distribution or the pseudo-priors are changed.

\begin{figure*}
        \centering
        \begin{subfigure}[b]{0.50\textwidth}
                \includegraphics[width=\textwidth]{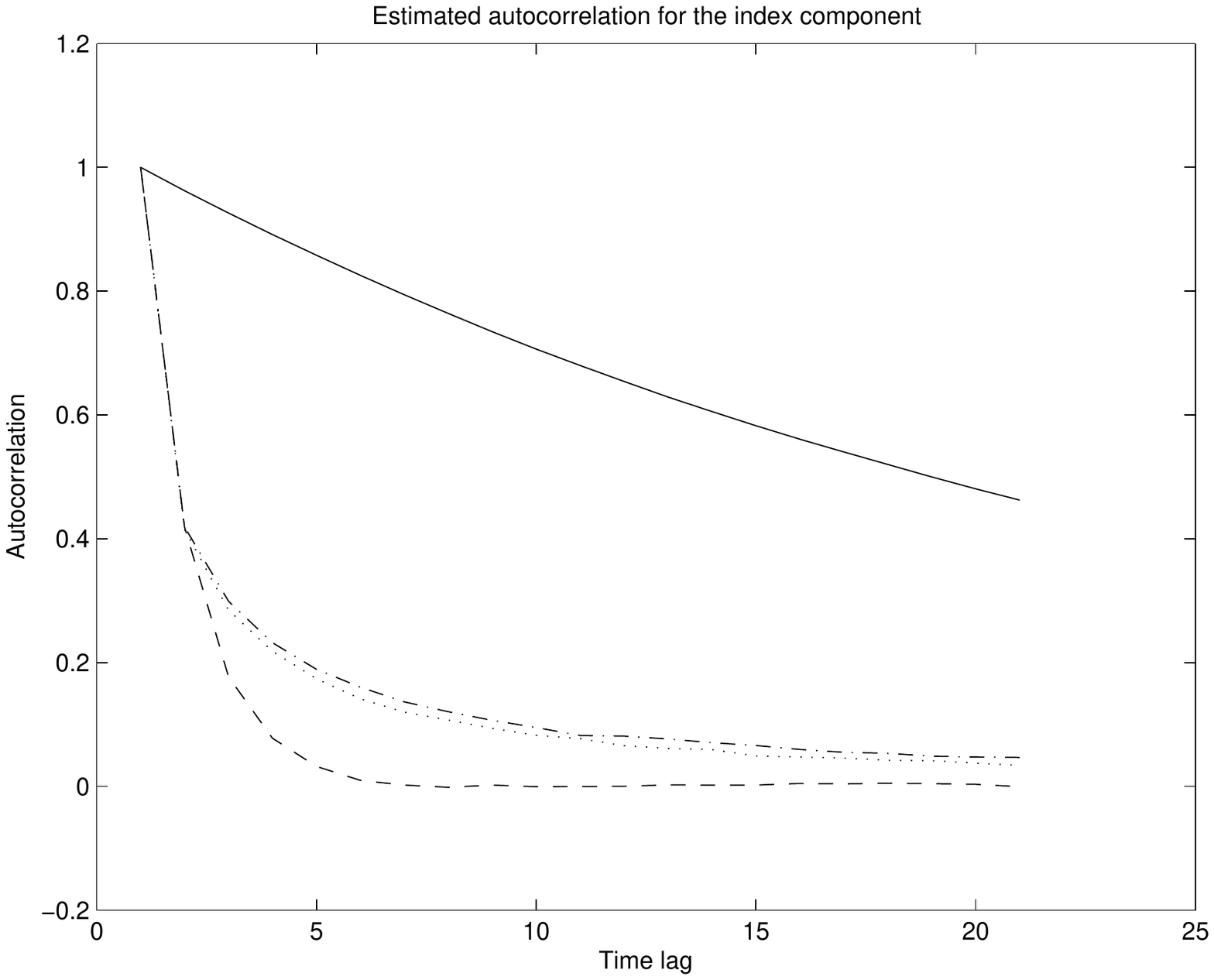}
                \caption{$M$-component}
                \label{fig:toy:m-component}
        \end{subfigure}~\begin{subfigure}[b]{0.50\textwidth}
                \includegraphics[width=\textwidth]{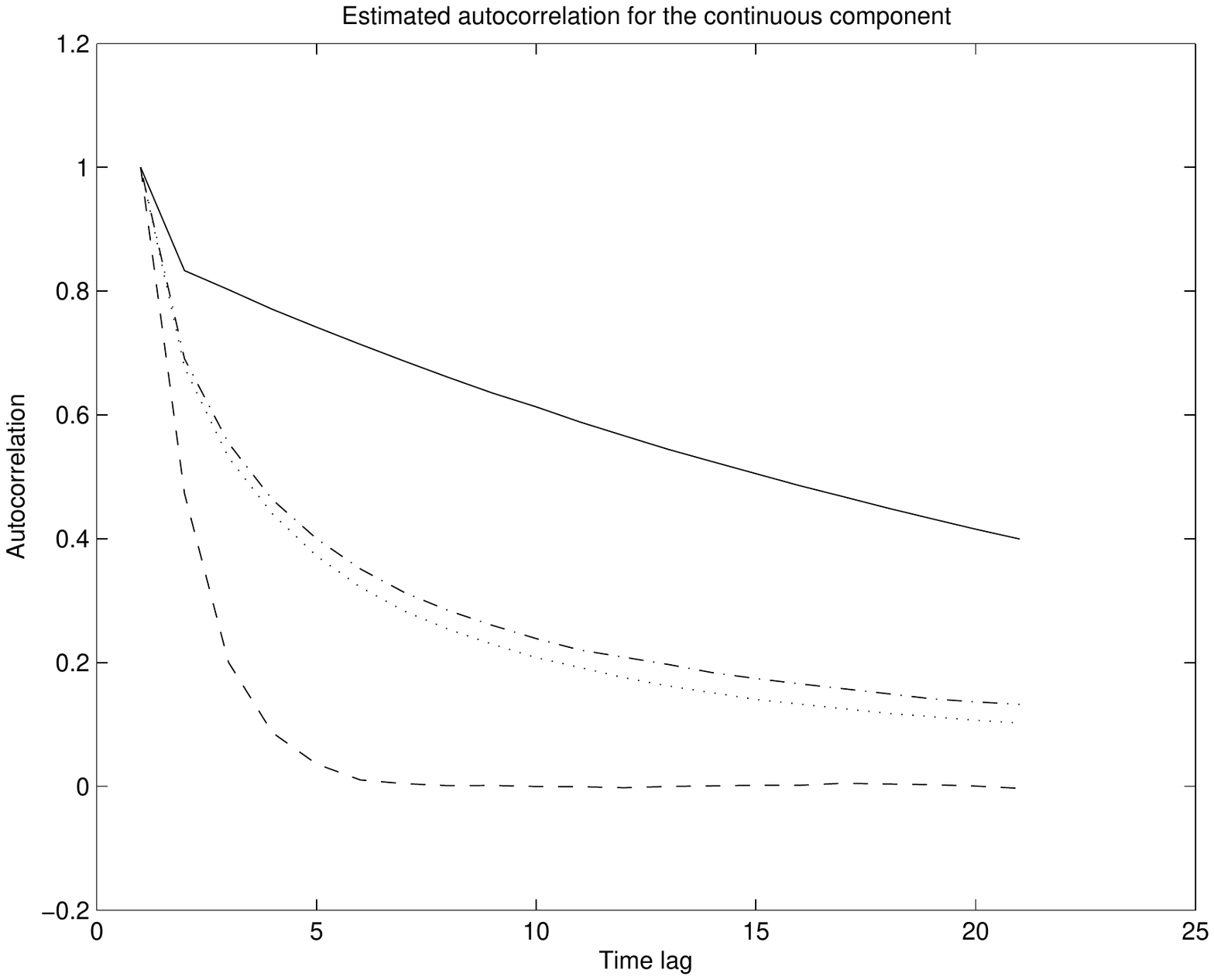}
                \caption{$Z$-component}
                \label{fig:toy:z-component}
        \end{subfigure}
        \caption{Plot of estimated autocorrelation for the standard Gibbs sampler (solid line), \autoref{alg:algCC} (dashed line), \autoref{alg:algMCC} (dotted line), and \autoref{alg:algFCC} (dash-dotted line) when applied to the model \eqref{eq:toy:model}.} \label{fig:autocorrelation:toy:model}
\end{figure*}

\subsection{Partially observed mixture variables}
\label{sec:partial:observations}

In this example we consider a model with two layers, where a pair $Y = (M, Z)$ of random variables, forming a mixture model $\nonobs$ on $\Ysp = \intvect{1}{n} \times \Zsp$ of the form described in \autoref{subsec:carlinchib}, is only \emph{partially observed} through some random variable $X$ taking values in some other state space $(\Xsp, \Xalg)$. More specifically, we assume that the distribution of $X$ conditionally on $Y$ is given by some Markov transition density $\Lhd$ on $\Ysp \times \Xsp$, i.e.
$$
X \mid (M, Z) \sim \Lhd((M, Z), x) \refx(\rmd x) \eqsp,
$$
where $\refx$ is some reference measure on $(\Xsp, \Xalg)$. When operating on a model of this form one is typically interested in computing the conditional distribution of the latent variable $Y$ given some distinguished value $X = x \in \Xsp$ of the observed variable. This posterior distribution has the density
\begin{multline*}
\starg(\m, \z \mid x) = \frac{\Lhd((\m, \z), x) \nonobs(\m, \z)}{\iint \Lhd((\m, \z), x) \nonobs(\m, \z) |\rmd \m|Ê\nu(\rmd \z)} \\Ê((\m, \z, x) \in \Ysp \times \Xsp)
\end{multline*}
w.r.t. the product $|\rmd \m|Ê\nu(\rmd \z)$. Since the observation $x$ is fixed, we simply omit this quantity from the notation and write $\starg(\m, \z \mid x) = \starg(\m, \z)$. Note that $\starg$ is again a mixture model on $\Ysp$, and our objective is to sample this distribution.

In order to evaluate, in this framework, the performances of the MCMC samplers discussed in the previous section we let, as in the previous example, $\Ysp = \intvect{1}{2} \times \rset$ and consider the Gaussian mixture model
$$
\nonobs(\m, \z) = \alpha_\m \normd{\z}{\mu_m}{\sigma^2} \quad ((\m, \z) \in \intvect{1}{2} \times \rset) \eqsp,
$$
where $\alpha_1 = 1/4$, $\alpha_2 = 3/4$, $\mu_1 = -1$, $\mu_2 = 1$, and $\sigma = \sqrt{.2}$. (Note that letting $\alpha_1 = \alpha_2 = 1/2$ yields the mixture model \eqref{eq:toy:model} of the previous example.) In addition, we let be $(M, Z)$ be partially observed through
\begin{equation} \label{eq:posterior:model}
X = Z^2 + \varsigma \varepsilon \eqsp,
\end{equation}
where $\varsigma = \sqrt{.1}$ and $\varepsilon$ is a standard Gaussian noise variable which is independent of $Z$. Consequently, the measurement density (with respect to Lebesgue measure) is given by $\Lhd((\m, \z), x) = \normd{x}{\z^2}{\varsigma^2}$, $x \in \rset$, in this case. For the fixed observation value $x = .4$ we estimated the posterior distribution
\begin{multline*}
\starg(\m, \z) \propto \alpha_\m \normd{\z}{\mu_m}{\sigma^2} \normd{x}{\z^2}{\varsigma^2} \\
((\m, \z) \in \intvect{1}{2} \times \rset)
\end{multline*}
and the corresponding posterior mean $\mu_z \eqdef \int z \starg(\rmd \z)$
using \autoref{alg:algMCC} and \autoref{alg:algFCC}. Note that we are unable to sample directly the conditional distribution $\starg(\z \mid \m)$ in this case due to the nonlinearity of the observation equation \eqref{eq:posterior:model}; thus, \autoref{alg:algCC} is excluded from our comparison. In addition, we implemented the Gibbs sampler \autoref{alg:algGibbs} with Step~(ii) replaced by a Metropolis-Hastings operation, yielding a Metropolis-within-Gibbs (MwG) sampler. This Metropolis-Hastings operation as well as in the corresponding operation in Step~(iii) of the MCC sampler (\autoref{alg:algMCC}) used the conditional prior distribution as proposal, e.g.
$$
R_\ell(\auxi, \rmd \z) = \nonobs(\rmd \z \mid \m) \quad ((\ell, \auxi) \in \intvect{1}{2} \times \rset) \eqsp.
$$
This distribution was also used for designing the pseudo-priors in the MCC and FCC algorithms, e.g.,
$$
\rho_\ell(\rmd \z) = \nonobs(\rmd \z \mid \m) \quad ((\ell, \auxi) \in \intvect{1}{2} \times \rset) \eqsp,
$$
and consequently the MCC sampler can, as in the previous example, be viewed as a ``random refreshment''-version (using the terminology of \cite{maire:douc:olsson:2014}) of the FCC sampler. The resulting autocorrelation function estimates are displayed in \autoref{fig:autocorrelation:posterior:model}, which shows that the FCC and MCC algorithms are clearly superior, in terms of autocorrelation, to the MwG sampler. Even though the MCC sampler has, as expected from \autoref{thm:compMCC-FCC}, a small advantage to the FCC sampler in terms of autocorrelation, both samplers exhibit, very similar mixing properties. This is particularly appealing in the light of the CPU times reported in \autoref{tab:mean:CPU}, which shows that the FCC sampler almost to twice as fast as the MCC sampler for our implementation.  \autoref{tab:mean:CPU} reports also the posterior mean estimates obtained with the different algorithms, and apparently the slow mixing of the MwG sampler rubs off on the precision of the corresponding estimate. The true value, $\mu_z = .315$, was obtained using numerical integration.

\begin{figure*}
        \centering
        \begin{subfigure}[b]{0.50\textwidth}
                \includegraphics[width=\textwidth]{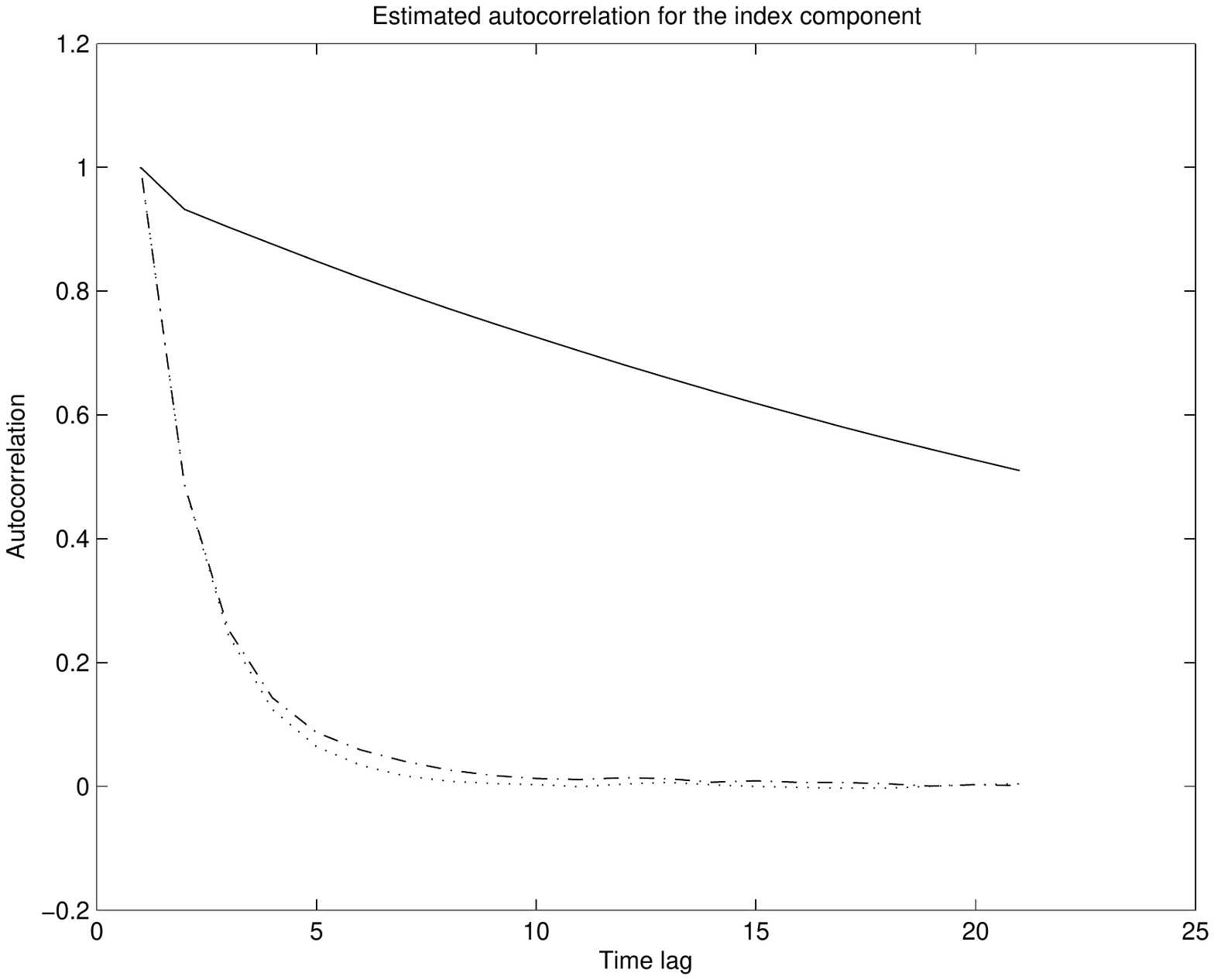}
                \caption{$M$-component}
                \label{fig:posterior:m-component}
        \end{subfigure}~\begin{subfigure}[b]{0.50\textwidth}
                \includegraphics[width=\textwidth]{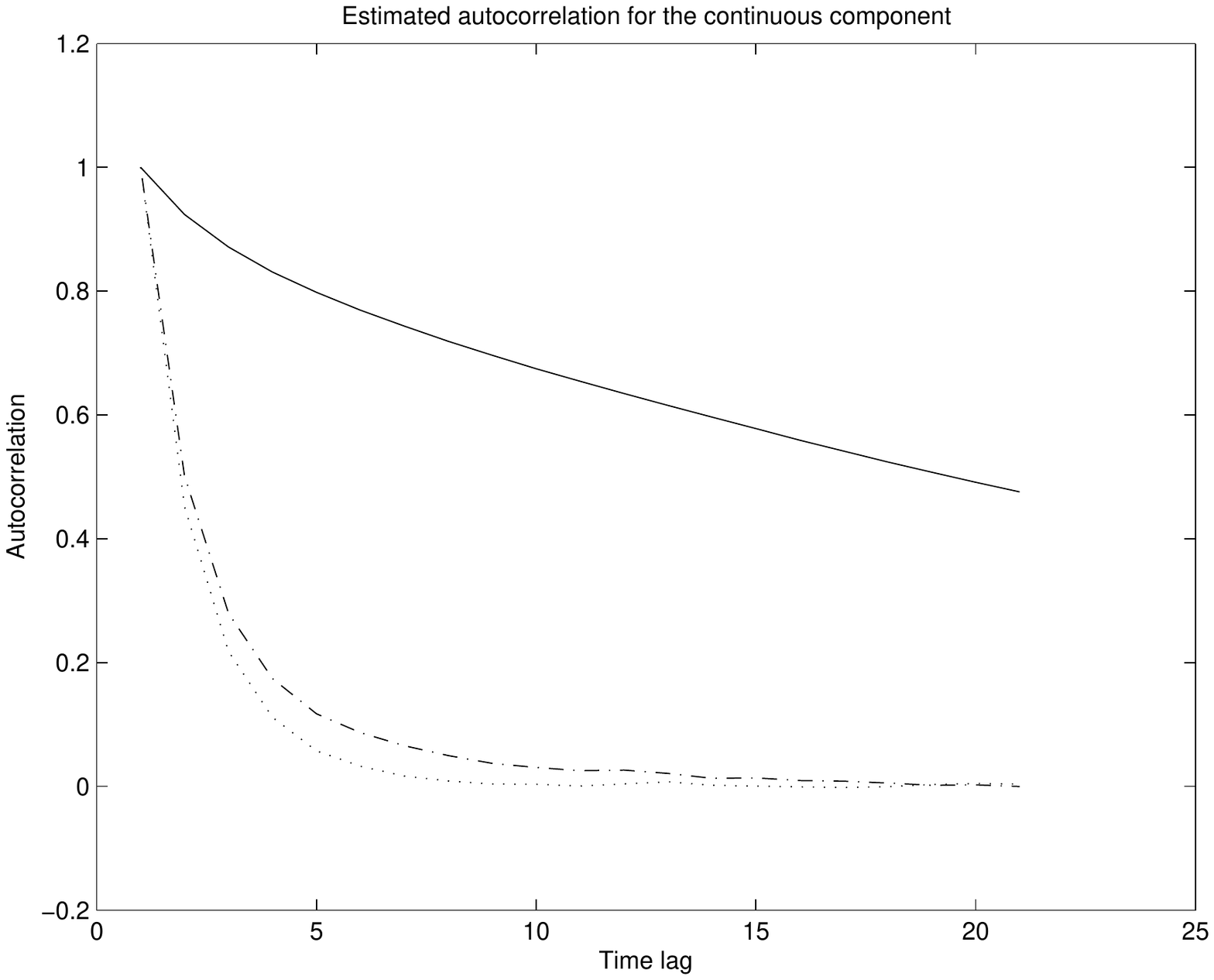}
                \caption{$Z$-component}
                \label{fig:posterior:z-component}
        \end{subfigure}
        \caption{Plot of estimated autocorrelation for the Metropolis-within-Gibbs sampler (solid line), \autoref{alg:algMCC} (dotted line), and \autoref{alg:algFCC} (dash-dotted line) when applied to the model \eqref{eq:posterior:model}.} \label{fig:autocorrelation:posterior:model}
\end{figure*}

\begin{table}
\caption{Posterior means delivered the MwG, MCC, and FCC algorithms for the partially observed mixture model \eqref{eq:posterior:model} together with the corresponding CPU times. The true posterior mean (for $x = .4$) is $\mu_z = .315$.}
\label{tab:mean:CPU}
\begin{tabular}{c|ccc}
\hline\noalign{\smallskip}
algorithm & mean & CPU time (s) \\
\noalign{\smallskip}\hline\noalign{\smallskip}
MwG & $.334$ & $50.9$ \\
MCC & $.311$ & $58.7$ \\
FCC & $.314$ & $33.4$ \\
\noalign{\smallskip}\hline
\end{tabular}
\end{table}

\autoref{fig:density:estimate} displays the estimate of the marginal posterior density $\starg(\z)$ obtained by applying a Gaussian kernel smoothing function to the output of the FCC algorithm. The exact posterior, obtained using numerical integration, is plotted for comparison.

Finally, we remark that also the results obtained in this example appear to be relatively insensitive to the parametrisation of the model and the pseudo-priors.

\begin{figure*}
\centering
\includegraphics[width=0.50\textwidth]{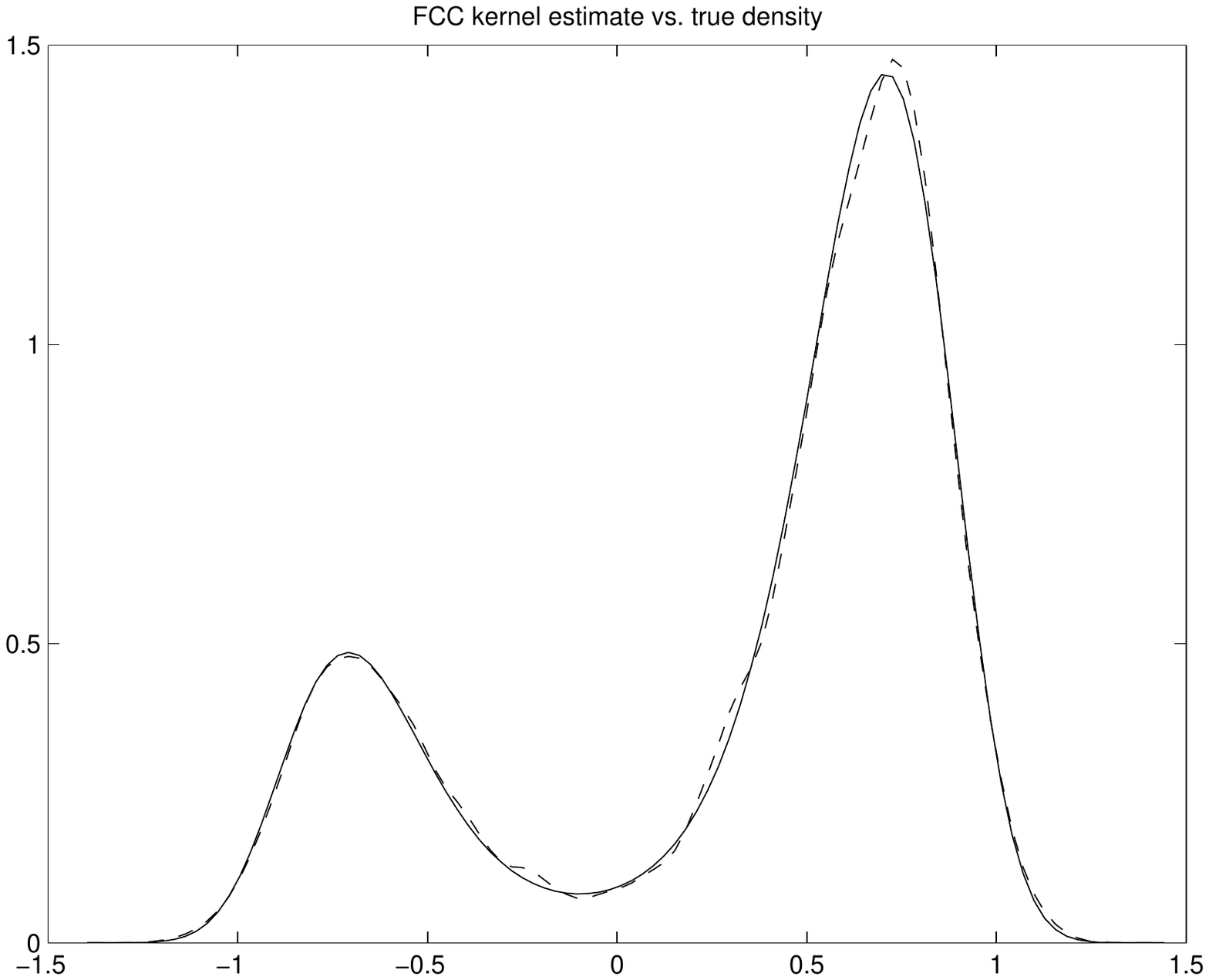}
\caption{Probability density estimate based on the sequence $\{ \Z[k]{\FCC} ; k \in \intvect{1001}{10^5} \}$ generated by \autoref{alg:algFCC} (dashed line)  for the partially observed mixture model \eqref{eq:posterior:model} together with the exact posterior density (solid line).}
\label{fig:density:estimate}       % Give a unique label
\end{figure*}

\section{Conclusion}
\label{sec:conclusion}
We have compared some data-augmentation-type MCMC algorithms sampling from mixture models comprising a discrete as well as a continuous component. By casting Carlin \& Chib's pseudo-prior into our framework we obtained a sampling scheme that is considerably more efficient than the standard Gibbs sampler, which in general exhibits poor state-space exploration due to strong correlation between the discrete and continuous components (as a result of the highly multimodal nature of the mixture model). In the case where simulation of the continuous component $Z$ conditionally on $M$ is infeasible, we used a metropolised version of the algorithm, referred to as the MCC sampler, that handled this issue by means of an additional Metropolis-Hastings step in the spirit of the hybrid sampler. In this case our simulations indicate, interestingly, that the loss of mixing caused by simply passing, as in the FCC algorithm, the value of the $M$th auxiliary variable, generated by sampling from the pseudo-priors at the beginning of the loop, directly to $Z$ without any additional refreshment is limited. Thus, we consider the FCC algorithm, which we proved to be $\starg$-reversible, as strong contender to the MCC sampler in terms of efficiency (variance per unit CPU).  

Our theoretical results comparing the MCC and FCC samplers deal exclusively with mixing properties of the restriction of the MCMC output to the discrete component, and the extension of these results to the continuous component is left as an open problem. However, we believe that the discrete component is indeed the quantity of interest, as our simulations indicate that the degree mixing of the discrete component gives a limitation of the degree of mixing of the bivariate chain due to the multimodal nature of the mixture. 

There are several possible improvements of the FCC algorithm. For instance, following \cite{petralias:mcmc}, only a subset of the pseudo-priors (namely those with indices belonging to some neighborhood of the current $M$) could be sampled at each iteration, yielding a very efficient algorithm from a computational point of view. Such an approach could be also used for handling the case of an infinitely large index space (i.e. $n = \infty$).

%BibTeX users please use one of
\bibliographystyle{spmpsci}      % mathematics and physical sciences
%\bibliography{../../AoS/Round2/biblio}   % name your BibTeX data base
\bibliography{biblio}   % name your BibTeX data base

%%
%% Non-BibTeX users please use
%\begin{thebibliography}{}
%%
%% and use \bibitem to create references. Consult the Instructions
%% for authors for reference list style.
%\end{thebibliography}

\end{document}